\documentclass[superscriptaddress,aps,pra,twocolumn,showpacs,nofootinbib,longbibliography]{revtex4-2}
\usepackage{amsmath,amssymb,amsthm}
\usepackage{easybmat}
\usepackage[colorlinks=true,citecolor=blue,urlcolor=blue]{hyperref}
\usepackage{times,txfonts}
\usepackage{braket}
\usepackage{color,soul}
\usepackage{natbib}
\usepackage{subcaption}
\usepackage{graphicx}
\usepackage[normalem]{ulem}
\usepackage{slashbox}

\newcommand{\be}{\begin{equation}}
\newcommand{\ee}{\end{equation}}
\newcommand{\ba}{\begin{eqnarray}}
\newcommand{\ea}{\end{eqnarray}}
\newcommand{\ketbra}[2]{|#1\rangle \langle #2|}
\newcommand{\tr}{\operatorname{Tr}}

\newcommand{\etal}{{\it{et. al. }}}

\newtheorem{cor}{Corollary}

\newtheorem{proposition}{Proposition}

\newcommand{\I}{1\!\!1}

\begin{document}
 \newwrite\bibnotes
    \def\bibnotesext{Notes.bib}
    \immediate\openout\bibnotes=\jobname\bibnotesext
    \immediate\write\bibnotes{@CONTROL{REVTEX41Control}}
    \immediate\write\bibnotes{@CONTROL{%
    apsrev41Control,author="08",editor="1",pages="1",title="0",year="1"}}
     \if@filesw
     \immediate\write\@auxout{\string\citation{apsrev41Control}}%
    \fi
\title{Certifying quantumness beyond steering and nonlocality and its implications on  quantum information processing}  
\author{Chellasamy Jebarathinam}
  \email{jebarathinam@gmail.com}
 \affiliation{Department of Physics and Center for Quantum Frontiers of Research \& Technology (QFort), National Cheng Kung University, Tainan 701, Taiwan}
 \affiliation{Center for Theoretical Physics, Polish Academy of Sciences, Aleja Lotnik\'ow 32/46, 02-668 Warsaw, Poland}

 \affiliation{Department of Physics and Center for Quantum Information Science, National Cheng Kung University, Tainan 701, Taiwan}
 
 \author{Debarshi Das}
  \email{dasdebarshi90@gmail.com}
\affiliation{Department of Physics and Astronomy, University College London, Gower Street, WC1E 6BT London, England, United Kingdom}

\begin{abstract}
Superunsteerability is a particular kind of spatial quantum correlation that can be observed in a steering scenario in the presence of limited shared randomness. 
In this work, we define an experimentally measurable quantity in a steering scenario to certify superunsteerability. In the context of certification of randomness with this scenario, we demonstrate that such certification of superunsteerability  provides a bound on the amount of genuine randomness generation.  On the other hand, superlocality is another kind of spatial quantum correlation that can be observed in a Bell scenario in the presence of limited shared randomness. We identify inequalities to certify superlocality in the Bell scenarios that can be adopted to implement $2$-to-$1$ and $3$-to-$1$ random-access codes. We observe that such certification of superlocality acts as resource for the random-access codes in the presence of limited shared randomness. As a by-product of our certification of superunsteerability and superlocality, we identify a new classification of separable states having quantumness. 
\end{abstract}

\maketitle 
   \date{\today}
   
   \section{Introduction}

The fundamental aspects of quantum theory such as  quantum coherence and quantum nonlocality are valuable 
resources underpinning quantum technologies \cite{SAG+17,HHH+09,BCP+14,UCN+19}.  Therefore, characterizing these resources \cite{CG19} is  important 
for both the foundations of quantum mechanics and quantum information science. Recently, a relationship
between quantum coherence and quantum discord has been established \cite{MYG+16}.  Quantum discord  captures 
quantumness even in separable states \cite{OZ01,HV01} and  characterizes the quantum resources in certain quantum information processing tasks \cite{MBC+12,ATM16, BTD+17}.  Recently, it has been studied that even certain  separable states having quantum discord can be used to generate superlocal or superunsteerable correlations which
represent stronger than classical correlations in the presence of limited shared randomness (see Fig. (\ref{Fig:super}) and Fig. (\ref{Fig:superunsteer}))
\cite{DW15,CJ,JAS17,Jeb17,JDS+18,DBD+17,DAS201855}. Furthermore,  superunsteerability has been shown to be  a resource for quantum advantage in quantum information processing \cite{JDK+19,KJP+19}. 
Thus, quantum discord can be used as a resource for quantum tasks when the shared randomness is not a free resource.

\begin{figure}
	\begin{subfigure}{7cm}
		\centering\includegraphics[width=5cm]{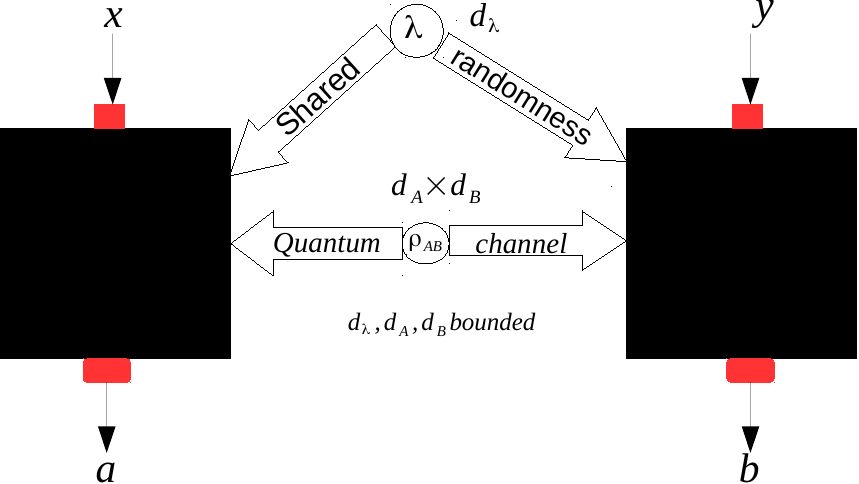}
		\caption{Superlocality   scenario where the dimension of resource (shared randomness
			$\lambda$ with the size $d_\lambda$
			or shared bipartite state $\rho_{AB}$ in $\mathbb{C}^{d_A}\otimes\mathbb{C}^{d_B}$) producing the box $\{p(a,b|x,y)\}$ is bounded. \label{Fig:super}}
	\end{subfigure}
	\begin{subfigure}{7cm}
		\centering\includegraphics[width=5cm]{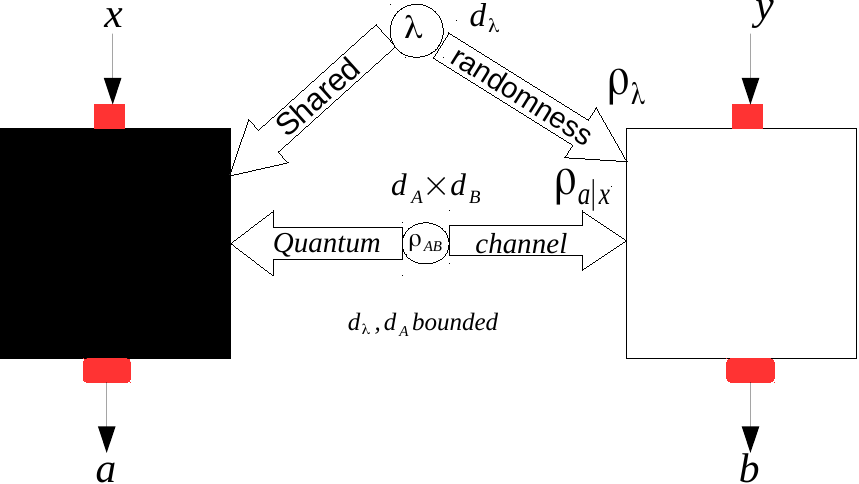}
		\caption{Superunsteerability  scenario where the dimension  $d_\lambda$ of the shared randomness
			$\lambda$ or local Hilbert-space dimension $d_A$ of  the shared bipartite state $\rho_{AB}$ in $\mathbb{C}^{d_A}\otimes\mathbb{C}^{d_B}$ producing the box $\{p(a,b|x,y)\}$ is bounded. \label{Fig:superunsteer}}
	\end{subfigure}

\end{figure}

Quantum coherence   manifests itself 
as a resource for genuine random number generation \cite{MYC+16,HGJ17}. 
In quantum random number generation protocols, genuine randomness is certified through 
the violation of a Bell inequality \cite{PAM+10} or uncertainty principle guarantees randomness \cite{VMT+14}.
In the approach based on Bell inequality, the protocol provides security in a device-independent way,
i.e., without the need to have knowledge about the internal working of the devices
while the approach based on the uncertainty principle requires trusted quantum devices.
Intermediate between these two approaches, an approach based on prepare and measure scenario 
was proposed to generate genuine randomness without the need to characterize the devices but by 
assuming only the Hilbert-space dimensions of the devices \cite{LBL+15}. This approach, where only the Hilbert-space dimensions of the devices are assumed, is known as semi-device-independent (SDI) approach.
SDI  quantum information protocols using shared quantum state as a quantum channel have also been studied  \cite{LTB+14,PCS+15,GBS16,Woo16}.

In the context of the prepare and measure protocols, analogous to the Bell inequalities,
dimension witnesses  have been derived for certifying the Hilbert-space dimension
\cite{GBH+10}. The dimension witnesses have been used as the  quantumness certification
for the semi-device-independent quantum information protocols \cite{PB11,LPY+12}. 
The relationship between the dimension witnesses and
the figure merit  of a quantum communication game called  quantum random-access codes \cite{Nay99}
enabled to demonstrate quantum random-access codes for secure quantum key distribution in a semi-device-independent way \cite{PB11}.

In Ref. \cite{PPK+09},  a task closely related to random-access codes which is implemented by Alice and Bob sharing 
a nonsignaling box  was studied to distinguish quantum and postquantum correlations.  
Subsequently, two of the authors of this work have demonstrated
implementation of  quantum random-access codes assisted by shared bipartite states replacing the quantum communication \cite{PZ10}. 
In this context,  correlations violating the suitable 
Bell inequalities assisted by one-bit of communication enable quantum advantage for the random-access codes.
The equivalence of implementation of quantum random-access codes using the prepare and measure scenario
and the scenario where a bipartite quantum state acts as quantum channel has led to study  the relationships between the dimension witnesses and
the Bell inequalities \cite{LMP+13}, 
the Popescu-Rohrlich boxes \cite{PR94} and random-access codes \cite{GHH+14,PS15}
and the sequential and spatial correlations used in the two types of scenarios for
implementing the quantum random-access codes \cite{TMP+16}.

Given the background that quantumness of the separable states is also useful for quantum 
information processing \cite{JDK+19,KJP+19}, it is important   to certify quantumness of the separable states
for self-testing of quantum information processing.
This  provides the motivation to construct semi-device-independent certification of 
quantumness beyond entanglement through witnessing superunsteerability of unsteerable 
correlations or 
superlocaltiy of local correlations.
To this end,  in this work, we  study  equivalence  of  quantumness of sequential correlations  in  a  prepare  and
measure scenario with independent devices and the quantumness of spatial correlations in the corresponding scenario with spatially separated correlated particles which realize the same task of the prepare and measure scenario.

Nonlinear dimension witnesses  constructed in Ref. \cite{BQB14}  serve
as the certification of quantumness of sequential correlations arising in the prepare and measure scenario 
in a semi-device-independent way.
In Ref. \cite{LBL+15}, such quantumness certification has been shown to achieve self-testing of
quantum random number generation. In the present study, we convert this prepare and measure protocol into the protocol
assisted by spatially separated bipartite quantum state. Furthermore, we 
extend the quantumness certification used in the prepare and measure scenario 
for the spatial correlations arising in this scenario involving spatially separated bipartite quantum state.
The quantumness certification defined for the spatial correlations can be used to provide
self-testing quantum random number generation.
We demonstrate that this quantumness certification witnesses 
superunsteerable correlations.
Thus,  the quantumness of the prepare and measure scenario gets manifested as superunsteerable correlations when the sequential correlations 
are observed as spatial correlations. 
Furthermore, we demonstrate a new classification of nonzero discord states based on the above certification of superunsteerability. Finally, we  address the connection of our certification of superunsteerability to certify superlocality.

We also consider the prepare and measure scenario  where the quantumness 
of the sequential correlations is observed in the presence of shared randomness
through a suitable linear dimension witness inequality \cite{GBH+10}. Such form of quantumness
has been exploited for achieving quantum random-access codes \cite{ALM+08}. 
The quantumness of sequential correlations in the presence of shared randomness gets manifested 
in the form of Bell nonlocal correlation in the corresponding scenario assisted by shared bipartite 
quantum state achieving the quantum random-access codes \cite{TMP+16}. We demonstrate that when certain sequential 
correlations which have quantumness and do not violate the dimension witness inequality are observed
as spatial correlations, they may become useful for implementing the quantum random-access codes
if there is a constraint on the amount of shared randomness. In this context, we point out that quantumness of the sequential correlations  gets manifested 
as superlocal correlations.

\section{Certifying superunsteerability}
\subsection{Prepare and measure scenario with independent devices} \label{P&MID}
\begin{figure}[h!]
	\begin{center}
		\includegraphics[width=7.5cm]{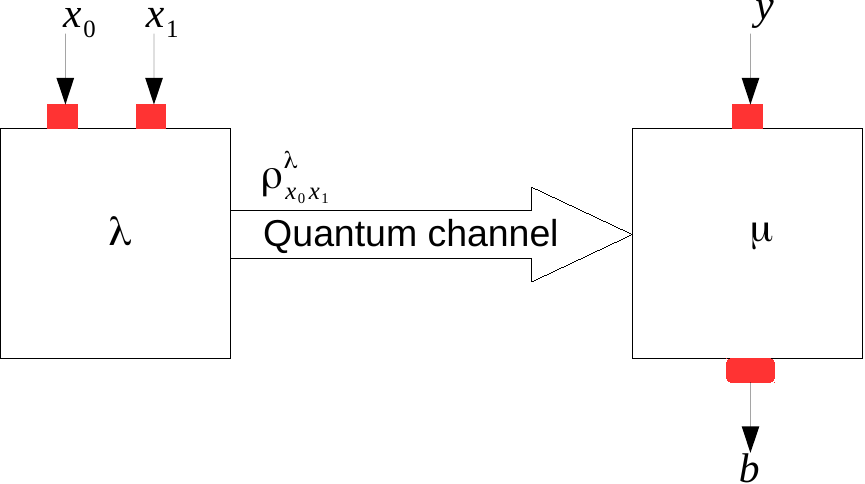}
	\end{center}
	\caption{Prepare and measure  scenario with uncorrelated devices. Here, $x_0,x_1,y,b \in \{0,1\}$ and $\rho^{\lambda}_{x_0x_1} \in \mathcal{B}(\mathbb{C}^{2})$, where $\mathcal{B}(\mathbb{C}^{2})$ is the set of bounded operators acting on the two-dimensional 
		Hilbert-space of quantum states $\mathbb{C}^{2}$.
		\label{Fig:PMNoSR}}
\end{figure}

In Ref. \cite{LBL+15}, a prepare and measure (P$\&$M) protocol was considered as shown in Fig. \ref{Fig:PMNoSR}  
for semi-device-independent (SDI) quantum random number generation (QRNG).
In this protocol, Alice has access to the set of four uncharacterized preparations labelled by  
$x_0x_1\in\{00,01,10,11\}$
which are in qubit states and sends one of them randomly to Bob through a quantum channel.  
Bob has access to the set of two uncharacterized measurements labelled by $y\in\{0,1\}$ and performs one of the measurements 
on the system received from Alice.  
To account for the imperfections, the devices are associated with 
an internal state denoted by $\lambda$ for Alice and $\mu$ for Bob. The devices are assumed to be independent, i.e.,  $p(\lambda, \mu)={q}_{\lambda }{r}_{\mu }$, $q_\lambda,r_\mu \ge 0 $, $\sum_{\lambda}{q}_{\lambda }=\sum_\mu{r}_{\mu }=1$. Bob observes the set of conditional probability of obtaining the outcomes $\{p(b|x_0x_1,y)\}$ which captures the correlations between Alice's preparations and Bob's measurements. The correlations are given by      
\begin{eqnarray}
	p(b|x_0x_1,y) & = & \sum _{\lambda ,\mu }{q}_{\lambda }{r}_{\mu }p(b|x_0x_1,y,\lambda ,\mu ) \nonumber \\ 
	& = & {\rm{Tr}}({\rho }_{x_0x_1}\frac{\I+(-1)^b{M}_{y}}{2}) \nonumber \\  
	& = & \frac{1}{2}(1+(-1)^b{\overrightarrow{S}_{x_0x_1}}\cdot {\overrightarrow{T}}_{y}), \label{BprobPM}
\end{eqnarray}
where  $\I$ is the $2 \times 2$ identity matrix and
\ba
{\rho }_{x_0x_1}&=&  \sum _{\lambda }{q}_{\lambda }{\rho }_{x_0 x_1}^{\lambda } =\frac{1}{2}(\I+{\overrightarrow{S}}_{x_0x_1}\cdot \overrightarrow{\sigma }),   \label{CSPM} \\
{M}_{y}&=&\sum _{\mu }{r}_{\mu }{M}_{y}^{\mu }={\overrightarrow{T}}_{y}\cdot \overrightarrow{\sigma }\mathrm{,} \label{BMop}
\ea
are the observed states  and measurements denoted in terms of the Bloch vectors $\overrightarrow{S}_{x_0x_1}$ and $\overrightarrow{T}_{y}$,
where $\overrightarrow{\sigma }$ is the vector of Pauli matrices.

To certify the quantumness of preparations and measurements in the above scenario, Lunghi \etal \cite{LBL+15} considered the following nonlinear dimension witness introduced in Ref. \cite{BQB14}:
\begin{align}
	W=\left|\begin{array}{cc}p(\mathrm{0|00,0})-p(\mathrm{0|01,0}) & p(\mathrm{0|10,0})-p(\mathrm{0|11,0})\\ 
		p(\mathrm{0|00,1})-p(\mathrm{0|01,1}) & p(\mathrm{0|10,1})-p(\mathrm{0|11,1})\end{array}\right|. \label{Wit}
\end{align}
The witness takes the value $0$ for any strategy involving a classical bit while it takes the value $0 \le W \le 1$ for any generic qubit strategy.  When the qubit preparations are classical (i.e., there exists a basis in which all states $\rho_{x_0 x_1}^{\lambda}$ are diagonal) or the measurements are classical (i.e., the measurements performed by Bob are commuting), the above witness $W$ takes the value $0$. Thus, $ 0 < W \le 1$ serves as certification of quantumness of the correlations $\{p(b|x_0x_1,y)\}$ in a SDI way since 
no assumption is required on the devices except that Alice and Bob prepare and measure on systems 
of the given dimension.
The witness takes the maximal value of $W=1$ for the following preparations:
\begin{align}
	\rho_{00}&=\frac{1}{2}\left(\I+\sigma_z \right)   \nonumber  \\
	\rho_{01}&=\frac{1}{2}\left(\I-\sigma_z \right)   \nonumber  \\
	\rho_{10}&=\frac{1}{2}\left(\I+\sigma_x\right)   \nonumber  \\
	\rho_{11}&=\frac{1}{2}\left(\I-\sigma_x\right)  
\end{align}
and measurements:
\begin{align}
	M_0&=\sigma_z, \quad     M_1=\sigma_x,
\end{align}
which correspond to the BB84 protocol \cite{BB84}.

\subsection{Relating the quantumness of sequential correlations with that of spatial correlations}

\begin{figure}[h!]
	\begin{center}
		\includegraphics[width=7.5cm]{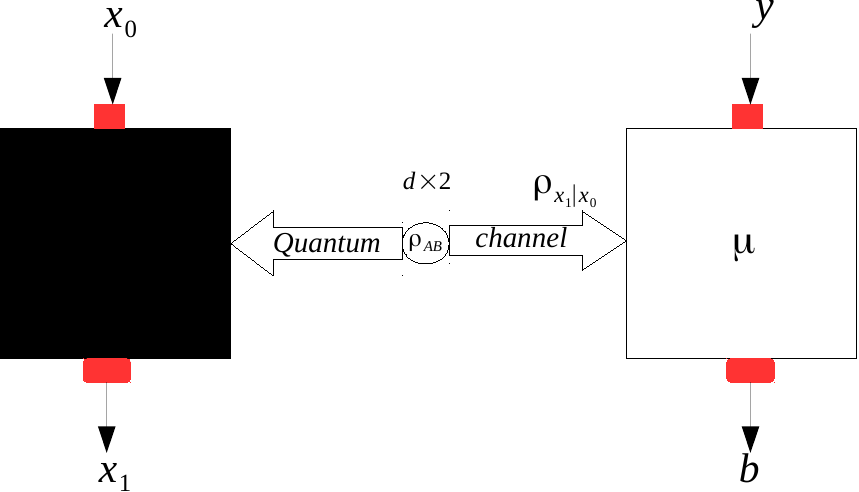}
	\end{center}
	\caption{A semi-device-independent scenario  where a shared bipartite state acts as the quantum channel to replicate the P$\&$M scenario in Fig. \ref{Fig:PMNoSR}.
		\label{Fig:SDI}}
\end{figure}
To observe the quantumness of  sequential correlations in the above scenario as spatial correlations,
let us consider a SDI scenario as shown in Fig.
\ref{Fig:SDI} where Alice and Bob share  a bipartite quantum state $\rho_{AB}$ of dimension $d \times 2$,  where $d$ is arbitrary. 
The spatial correlations that can be observed in this scenario can be seen to produce the sequential correlation in the P$\&$M scenario
as follows.
In the  SDI scenario as shown in Fig. \ref{Fig:SDI},  Alice has access to the set of two black-box measurements labelled by $x_0\in\{0,1\}$. The outcome of each measurement is denoted by $x_1\in\{0,1\}$.
These two measurements of Alice  prepare a set of four qubit states on Bob's side $\{\rho_{x_1|x_0}\}$.
The set of unnormalized conditional states prepared on Bob's side  $\{\sigma_{x_1|x_0}\}$ is given by 
\be
\sigma_{x_1|x_0}=\tr_A \left( M_{x_1|x_0} \otimes \I \rho_{AB}\right)=p(x_1|x_0)\rho_{x_1|x_0},
\ee
where  $M_{x_1|x_0}$ denotes the measurements operator of Alice's measurement and $p(x_1|x_0)$ is the probability of obtaining the outcome $x_1$ on Alice's side.
On the conditional states prepared by Alice's measurement, Bob performs one of two 
uncharacterized measurements labelled by $y\in \{0,1\}$ which is chosen randomly. 

Such a SDI 
scenario has been considered in Ref. \cite{MGH+16} for detecting dimension-bounded quantum steering where
only the Hilbert-space dimension of the trusted party is characterized and no other assumption, such as which measurements
the trusted party performs, is made.
This SDI scenario demonstrates quantum steering if $\{\sigma_{x_1|x_0}\}$
does not have a local-hidden-state (LHS) model \cite{WJD07, Pus13}.
The set of  normalized qubit states prepared on Bob's side in the 
Bloch representation is given by 
\be
\rho_{x_1|x_0}=\frac{1}{2}(\I+{\overrightarrow{S}}_{x_0x_1}\cdot \overrightarrow{\sigma }). \label{CSSDI}
\ee
Comparing Eqs. (\ref{CSPM}) and (\ref{CSSDI}), it follows that the sequential correlations 
produced in the P$\&$M scenario can be reproduced as spatial correlations in the SDI scenario as in Fig. \ref{Fig:SDI} 
with the following differences. 
In the case of scenario depicted in Fig. \ref{Fig:PMNoSR}, Alice uses two random bits as the input to the device,
on the other hand, in the case of scenario depicted in Fig. \ref{Fig:SDI}, Alice uses only one random bit as the input to the device. In the scenario depicted in Fig. \ref{Fig:SDI},
Bob receives the random preparations  which are the conditional states prepared by Alice's measurements
by a priory sharing a correlated quantum system with Alice, while
in the scenario depicted in Fig. \ref{Fig:PMNoSR}, Bob receives the random preparations through quantum communication.
In the case of scenario depicted in Fig. \ref{Fig:SDI}, the non-signaling
conditions from Alice to Bob \cite{Pus13} takes into account of all imperfections on the prepared states $\rho_{x_1|x_0}$
in a device-independent way, but the non-signaling condition is not satisfied in the SDI scenario as in Fig. \ref{Fig:PMNoSR}  as Alice can signal to Bob by sending the qubit. In the scenario depicted in Fig. \ref{Fig:PMNoSR},
we assume qubit dimension on Alice's side and associate an internal state $\lambda$ to account for 
noise/imperfection in the prepared states of the given dimension. On the other hand, in the SDI scenario depicted in Fig. \ref{Fig:SDI},
imperfection in Alice's device to prepare the relevant conditional states is taken into account in a device-independent way because of the nonsignaling conditions from Alice to Bob: 
\be
\sum_{x_1}\sigma_{x_1|0}=\sum_{x_1}\sigma_{x_1|1}, \quad \forall x_1
\ee
Here device-independence is applicable to Alice's device only while Bob's device works at the semi-device-independent level. Therefore, we do not need to associate an internal state $\lambda$ to
Alice's device to account for the imperfection in the conditional states prepared on Bob's side, but we associate an internal state $\mu$ on Bob's device to account for the imperfection in 
Bob's device.

In the scenario depicted in Fig. \ref{Fig:SDI}, as we model the imperfection in Bob's measuring device by an internal state $\mu$ 
which occurs with the probability $r_{\mu}$, $\sum_{\mu}r_{\mu}=1$, 
the spatial correlations as captured by the set of conditional probabilities $\{p(b|x_1;x_0,y)\}$ are given by   
\begin{eqnarray}
	p(b|x_1;x_0,y) & = & \sum _{\mu }{r}_{\mu }p(b|x_1;x_0 ,\mu ) \nonumber \\  
	& = & {\rm{Tr}}({\rho }_{x_1|x_0}\frac{\I+b{M}_{y}}{2}) \nonumber \\  
	& = & \frac{1}{2}(1+(-1)^b{\overrightarrow{S}}_{x_0x_1}\cdot {\overrightarrow{T}}_{y}), \label{BprobSDI}
\end{eqnarray}
where ${M}_{y}$ are given as in Eq. (\ref{BMop}).

To certify the quantumness of the preparations $\{\rho_{x_1|x_0}\}$ and measurements $\{M_y\}$ in the SDI scenario
depicted in Fig. \ref{Fig:SDI}, we extend the dimension witness given by Eq. (\ref{LDW}) for this scenario as follows:
\begin{align}
	Q=\left|\begin{array}{cc}p(\mathrm{0|0;0,0})-p(\mathrm{0|1;0,0}) & p(\mathrm{0|0;1,0})-p(\mathrm{0|1;1,0})\\ 
		p(\mathrm{0|0;0,1})-p(\mathrm{0|1;0,1}) & p(\mathrm{0|0;1,1})-p(\mathrm{0|1;1,1})\end{array}\right|. \label{QC}
\end{align}
Let us now demonstrate that $Q=0$ for any set of joint behaviours $\{P(x_1,b|x_0,y)=P(b|x_1;x_0,y)P(x_1|x_0)\}$ which has the form
$P(x_1,b|x_0,y)=P(x_1|x_0)P(b|y)$ (i.e., there is no correlation between the outcomes $x_1$ and $b$). 
For any set of joint behaviours which does not have correlation,
the conditional probability of observing Bob's outcome takes the form  $P(b|x_1;x_0,y)=P(b|y)$.
With this form of $P(b|x_1;x_0,y)$, evaluating $Q$ as given by Eq. (\ref{QC}) gives $Q=\braket{M_0}\braket{M_1}-\braket{M_1}\braket{M_0}$,
where $\braket{M_y}=P(0|y)-P(1|y)$.
This implies that $Q=0$ for the set of joint behaviours which do not have correlation. 
On the other hand, $Q>0$, in general, even for the unsteerable correlations. 
However, a nonzero $Q$ for the set of conditional probability distributions observed on Bob's side $\{p(b|x_1;x_0,y)\}$ implies that both  the preparations and measurements  have quantumness as we demonstrate now.
Comparing Eqs. (\ref{BprobPM}) and (\ref{BprobSDI}), it follows that any set of preparations and measurements giving rise to nonzero $Q$
in the scenario depicted in Fig. \ref{Fig:SDI} also gives rise to nonzero $W$ in the scenario depicted in Fig. \ref{Fig:PMNoSR} and vice versa.

\subsection{Certifying randomness with spatial correlations}
Let us now consider certification of randomness on Bob's side in the scenario 
depicted in Fig. \ref{Fig:SDI}. This scenario for randomness certification is similar to the random certification in a one-sided device-independent way analysed in Ref. \cite{LTB+14} with the following difference. In our context, we only assume the dimension on Bob's side, whereas in the context of Ref. \cite{LTB+14},
both the dimension and measurements are trusted on the side where randomness is certified.
Note that a nonzero value of $W$ given by Eq. (\ref{Wit}) implies randomness certification in the SDI scenario depicted in Fig. \ref{Fig:PMNoSR}.
Therefore, any set of conditional probability distributions  observed on Bob's side $\{p(b|a;x,y)\}$ in the scenario depicted in Fig. \ref{Fig:SDI} has 
intrinsic randomness if it gives rise to $Q>0$.   

To quantify the amount of randomness generated in our SDI scenario \ref{Fig:SDI}, let us define the average 
guessing probability of the events $\{x_0,x_1,y,b\}$  as 
${p}^{g}:=\frac{1}{4}\sum _{x_0,x_1,y,\mu }{r}_{\mu }\mathop{{\rm{\max }}}\limits_{b}p(b|x_1;x_0,y,\mu )$.
Then the randomness rate in our protocol can be quantified by the min-entropy of ${p}^{g}$.
By adopting the procedure given in Ref. \cite{FZW+17} which has been used to provide a tighter upper bound on the guessing probability, 
one can  derive the following upper bound on the guessing probability in our SDI scenario too:
\be \label{STQRNG}
\begin{array}{rcl}{p}^{g} & = & \frac{1}{4}\sum _{x_0,x_1,y ,\mu }{r}_{\mu }\mathop{{\rm{\max }}}\limits_{b}p(b|x_1;x_0,y,\mu )\\  &  & \le \,\frac{1}{2}(1+\sqrt{\frac{2-Q}{2}})\equiv f (Q) \end{array}
\ee
Thus, the min-entropy has a lower bound as $H_{\min}=-\log_2p_g=-\log_2f (Q) $ which implies that  $H_{\min}$ is a monotonic function of 
the quantity $Q$ which quantifies the quantumness  of the spatial correlations in our SDI scenario depicted in Fig. \ref{Fig:SDI}.
In the following, we identify what kind of quantum correlation  can act as a resource for
randomness certification in our SDI scenario.

\subsection{Quantum correlation beyond
	steering as a resource for certifying randomness}

In the context of a steering scenario where Alice performs a set of black-box measurements labelled by $x$
and Bob performs a set of quantum measurements of fixed dimension denoted by $M_y$, 
the set of joint behaviours  $\{p(a,b|x,y)\}$ detecting steerability from Alice to Bob
does not have a decomposition of the form given by  \cite{WJD07},
\be
p(a,b|x,y)= \sum^{d_\lambda-1}_{\lambda=0} p_{\lambda} p(a|x,\lambda) p(b|y; \rho_\lambda). \label{LHV-LHS}
\ee
Here, $\{p(a|x,\lambda)\}_{a,x}$ is the set of arbitrary probability distributions $p(a|x,\lambda)$ conditioned upon shared randomness/hidden variable $\lambda$ occurring with the probability $p_{\lambda}$; $\sum^{d_\lambda-1}_{\lambda=0} p_{\lambda} =1$. On the other hand, $\{p(b|y; \rho_\lambda) \}_{b,y}$ is the set of quantum probability distributions $p(b|y; \rho_\lambda)=\tr(\Pi_{b|y}\rho_\lambda)$, arising from some local hidden state $\rho_\lambda$ due to projective measurement associated with the projector $\Pi_{b|y}$, and
$d_\lambda$ denotes the dimension of the shared randomness/hidden variable $\lambda$.
The above decomposition is called a local hidden variable-local hidden state (LHV-LHS) model.
If the correlation arising from the given steering scenario does not have steerability, it can still have  quantumness in the form of superunsteerability  \cite{DBD+17,DAS201855}
when there is a restriction on the amount of shared randomness.

\vspace*{12pt}
\noindent
{\bf Definition~1.}
Superunsteerability  \cite{DBD+17,DAS201855} is defined as the requirement for a larger dimension of the classical variable that the steering party (Alice) preshares with the trusted party (Bob) for simulating the given unsteerable correlations, than that of the quantum state which reproduces them.
Suppose we have a quantum state in $\mathbb{C}^{d_A}\otimes\mathbb{C}^{d_B}$
and measurements which produce a unsteerable bipartite  box $\{ p(a,b|x,y) \}_{a,x,b,y}$.
Then, superunsteerability holds if and only if (iff) there is no decomposition of the box in the form given by Eq. (\ref{LHV-LHS}),
with  $d_\lambda\le d_A$.
\vspace*{12pt}
\noindent

Superunsteerability provides an operational characterization to the quantumness of unsteerable  boxes \cite{DBD+17}. Note that superunsteerability is defined in the standard steering scenario. The only difference is that there is a constraint on the dimension of the resources producing the correlations in the context of superunsteerability. 

Here we should emphasize that the above definition of superunsteerability is also applicable in the `dimension-bounded quantum steering scenario, where Bob's measurements are uncharacterized POVMs acting on a fixed dimensional Hilbert space. Hence, in this case, the probability distributions $p(b|y; \rho_\lambda)$ appeared in Eq.(\ref{LHV-LHS}) is given by, $p(b|y; \rho_\lambda)=\tr(M_{b|y}\rho_\lambda)$, where $M_{b|y}$s are POVM elements with $M_{b|y} \geq 0$ $\forall$ $b,y$ and $\sum_{b} M_{b|y} = \mathbb{I}$ $\forall$ $y$.

We now demonstrate the following proposition.
\vspace*{12pt}
\noindent
\begin{proposition}
	Suppose any nonzero value of the quantity $Q$  given by Eq. (\ref{QC})  arises from a two-qubit state. Then it implies the presence of  superunsteerability or steering in the correlation $\{P(x_1,b|x_0,y)\}$.
	\label{propo01}
\end{proposition}
\vspace*{12pt}
\noindent
{\bf Proof.} 
In the context of the SDI scenario depicted in Fig. \ref{Fig:SDI}, any two-qubit state which is either a classical-quantum (CQ) state,
\begin{equation}
	\rho_{CQ} = \sum_{i=0}^{1} p_i |i\rangle\langle i| \otimes \chi_i
	\label{eq:cq}
\end{equation}
or  a quantum-classical (QC) state,
\begin{equation}
	\label{cq:eq}
	\rho_{QC} = \sum_{j=0}^{1} p_{j}  \phi_j  \otimes | j \rangle \langle j|,
\end{equation}
with  $ \{  |  i  \rangle \}$  and  $  \{ |  j  \rangle  \}$ being
orthonormal sets, and, $\chi_i$ and $\phi_j$ being arbitrary quantum
states,
cannot be used to demonstrate superunsteerability or steering \cite{DBD+17}. Next, it will be demonstrated that 
any such two-qubit state  cannot
be used to give rise to nonzero value of the quantity $Q$ given by Eq. (\ref{QC}).

The classical-quantum states given by Eq. (\ref{eq:cq}) can be decomposed as follows
\ba
\rho_{CQ}&=&\frac{1}{4}\big[ \I \otimes \I + (p_0-p_1) \hat{r} \cdot \vec{\sigma} \otimes \I + \I \otimes \left(p_0 \vec{s}_0+  p_1 \vec{s}_1 \right)\cdot \vec{\sigma} \nonumber \\
&+&\hat{r} \cdot \sigma \otimes \left( p_0\vec{s}_0-p_1 \vec{s}_1 \right)\cdot \vec{\sigma} \big],
\ea
where $\hat{r}$ is the Bloch vector of the projectors $\ketbra{i}{i}$ and  $\vec{s}_i$ are the Bloch vectors of the quantum states $\chi_i$. Since we only make qubit assumption on Alice's side, her measurements are \textit{a priory}  POVM
with elements given by
\be
M_{x_1|x_0}=\gamma^{x_1}_{x_0} \I+ (-1)^{x_1} \frac{\eta_{x_0}}{2} \hat{u}_{x_0} \cdot \vec{\sigma}, \label{POVMalice}
\ee
where $x_0 \in \{0,1\}$, $x_1 \in \{0,1\}$, $\gamma^{0}_{x_0}+\gamma^{1}_{x_0}=1$ $\forall$ $x_0$ and $0 \leq \gamma^{x_1}_{x_0} \pm \frac{\eta_{x_0}}{2} \leq 1$ $\forall$ $x_0, x_1$.
On the other hand, Bob's
measurements are  also  POVMs
with elements given by
\be
M_{b|y}=\gamma^{b}_{y} \I+ (-1)^{b} \frac{\eta_{y}}{2} \hat{u}_{y} \cdot \vec{\sigma}, \label{POVMbob}
\ee
where $y \in \{0,1\}$, $b \in \{0,1\}$, $\gamma^{0}_{y}+\gamma^{1}_{y}=1$ $\forall$ $y$ and $0 \leq \gamma^{b}_{y} \pm \frac{\eta_{y}}{2} \leq 1$ $\forall$ $y, b$. Now, It can be easily checked that the above measurement settings always lead to $Q=0$ for the classical-quantum states $\rho_{CQ}$.

The quantum-classical states given by Eq. (\ref{eq:cq}) can be decomposed as follows
\ba
\rho_{QC}&=&\frac{1}{4}\big[ \I \otimes \I + \left(p_0 \vec{s}_0+  p_1 \vec{s}_1 \right)\cdot \vec{\sigma} \otimes  \I + \I \otimes (p_0-p_1) \hat{r} \cdot \vec{\sigma}  \nonumber \\
&+&  \left( p_0\vec{s}_0-p_1 \vec{s}_1 \right)\cdot \vec{\sigma}  \otimes  \hat{r} \cdot \sigma \big],
\ea
where    $\vec{s}_i$ are the Bloch vectors of the quantum states $\phi_j$ and $\hat{r}$ is the Bloch vector of the projectors $\ketbra{j}{j}$. For the measurement operators given  by Eqs. (\ref{POVMalice}) and (\ref{POVMbob}), the above quantum-classical states always give $Q=0$.

Note that in our SDI scenario, 
any unsteerable box having the decomposition (\ref{LHV-LHS}) can be reproduced by a classical-quantum state of
the form $\sum^{d_\lambda-1}_{\lambda=0} p_\lambda \ketbra{\lambda}{\lambda} \otimes \rho_\lambda$ \cite{DBD+17},
where $\{\ketbra{\lambda}{\lambda}\}$ forms an orthonormal basis in  $\mathbb{C}^{d_\lambda}$, with $d_\lambda \le 4$ \footnote{Here the dimension of the hidden variable is upper bounded by $4$
	since any local as well as unsteerable correlation corresponding to this scenario can be simulated by shared classical randomness of dimension $d_\lambda \le 4$ \cite{DW15,DBD+17}.}.
This implies that any unsteerable box produced from a two-qubit state that requires a hidden variable of dimension $d_\lambda =2$
for providing a LHV-LHS model (i.e., any unsteerable box produced from a two-qubit state that is not superunsteerable)
can be simulated by a two-qubit state which admits the form of the classical-quantum  state given by Eq.(\ref{eq:cq}).
Thus, for any such unsteerable box,
$Q=0$. On the other hand, $Q>0$ produced from a two-qubit state certifies that the box does not arise from a classical-quantum or quantum-classical state. Hence, that box either
requires the hidden variable of dimension $d_\lambda>2$ for providing a LHV-LHS model (i.e., the box is superunsteerable) or that box is steerable
$\square\,$.

Now, in order to illustrate proposition \ref{propo01}, let us give an example where non-zero $Q$ arising from two-qubit state certifies superunsteerability or steering. Consider the following white-noise BB84 family,
\begin{equation}
	\label{bb84}
	P_{BB84}(x_1 b|x_0 y) = \frac{1 + (-1)^{x_1 \oplus b \oplus x_0.y} \delta_{x_0,y} V }{4},
\end{equation}
where $V$ is a real number such that $0 < V \leq 1$; $x_0$, $y$ denote the input variables on Alice's and Bob's sides respectively; and $x_1$, $b$ denote the outputs on Alice's and Bob's sides respectively, $x, y, a, b \in \{0, 1\}$. This family of correlations can be produced from the two-qubit Werner state,
\begin{equation}
	\label{w}
	\rho_V = V | \psi^- \rangle \langle \psi^-| + \frac{1-V}{4} \I \otimes \I,
\end{equation}
with $|\psi^- \rangle = \frac{1}{\sqrt{2}} (|01 \rangle - |10 \rangle)$  and $0 < V \leq 1$ if Alice performs the projective measurements of observables corresponding to the operators $A_0 = - \sigma_z$ and $A_1 = \sigma_x$, and Bob performs projective measurements of observables corresponding to the operators $B_0 =  \sigma_z$ and $B_1 = \sigma_x$ \cite{DBD+17}. Note that, for $V=1$, the correlation (\ref{bb84}) corresponds to the quantum key distribution protocol of Ref. \cite{BBM92}.

The above family of correlations is superunsteerable for $0< V \leq \frac{1}{\sqrt{2}}$ and steerable for $\frac{1}{\sqrt{2}} < V \leq 1$ \cite{DBD+17}. It can be easily checked that for the white-noise BB84 family given by Eq.(\ref{bb84}), $Q = V^2$. Hence, for this family of correlations, $Q>0$ for $V>0$. Non-zero value of $Q$ arising from the white-noise BB84 family, therefore, certifies superunsteerability or steering.

In the proof of proposition \ref{propo01}, we have demonstrated that any nonzero $Q$ certifies superunsteerability of two-qubit states 
without any assumption on the measurements.  Thus, we obtain the following corollary of this proposition.
\vspace*{12pt}
\noindent
\begin{cor}
	In case of unsteerable correlations,
	any nonzero value of the quantity $Q$ given by Eq. (\ref{QC}) certifies superunsteerability in a SDI way, i.e., 
	with the assumption of qubit dimension for each party and no assumption on the measurements. 
\end{cor}
\vspace*{12pt}
\noindent

The above proposition \ref{propo01} demonstrates that superunsteerability or steering is linked with non-zero values of $Q$ in the context of two-qubit states. As discussed earlier, the genuine randomness produced in our SDI scenario is quantified by $Q$. Hence, quantum correlation as captured by superunsteerability  or steering acts as a resource for genuine randomness 
generation in the SDI depicted in Fig. \ref{Fig:SDI}. Thus, the witness $Q$ provides self-testing processing of 
quantum information by certifying the quantum phenomenon called superunsteerability/steering in a SDI way.

Note that any superunsteerable or steerable state must be a non-zero discord state (which cannot be expressed in classical-quantum/quantum-classical/classical-classical form) \cite{DBD+17}. Since, proposition \ref{propo01} demonstrates that any non-zero value of $Q$ arising from any two-qubit state certifies superunsteerability or steering, it certifies non-zero discord as well. But the converse may not be true always.  Next, we will address which two-qubit  nonzero discord  states can be used to demonstrate nonzero $Q$. 

In Ref.  \cite{Luo08}  it was shown that any two-qubit state, up to local unitary transformations, can be reduced to the following form:
\be
\zeta=\frac{1}{4}\left(\I \otimes \I + \vec{a} \cdot \vec{\sigma} \otimes \I
+ \I \otimes \vec{b} \cdot \vec{\sigma} +\sum^3_{i=1} c_i \sigma_i \otimes \sigma_i \right), \label{can02q}
\ee
where  $\{\vec{a},\vec{b},\vec{c}\} \in \mathbf{R}^3$ are vectors with norm less than or equal to unity and $\vec{a}^2+\vec{b}^2+\vec{c}^2\le3$.
Let Alice's projective qubit measurements are given by the measurement operators
\be
M_{x_1|x_0}=\frac{1}{2}\left( \I+ (-1)^{x_1} \hat{u}_{x_0} \cdot \vec{\sigma} \right). \label{MOalice}
\ee
For such measurements, the conditional states prepared on Bob's side are given by
\be
\rho_{x_1|x_0}=\frac{1}{2}\left(\I+ \frac{\sum^3_{i=1}\left(b_i+(-1)^{x_1}{u_i}_{x_0}c_i\right)\sigma_i}{1+(-1)^{x_1}  \hat{u}_{x_0} \cdot \vec{a}}\right).
\ee

For simplicity, without lose of generality, let us consider the two-qubit states given by Eq. (\ref{can02q})
with $|c_1|\ge |c_2| \ge |c_3| $ and let Alice performs measurements along the directions $\hat{u}_{0}=\hat{x}$
and $\hat{u}_{1}=\hat{y}$ and Bob performs the measurement along the directions $\vec{T}_{0}=\hat{x}$
and $\vec{T}_{1}=\hat{y}$. Note that for this choice of measurement settings, the correlations arising from certain two-qubit states violate the two-setting linear steering inequality and the correlations arising from maximally entangled two-qubit states violate the two-setting linear steering inequality maximally \cite{CA16}. 
With this choice of the two-qubit states and measurements, $Q$ as given by Eq. (\ref{QC}) has been evaluated
to be of the form given by 
\be
Q=\frac{|(c_1-a_1b_1)(c_2-a_2b_2)-a_1b_1a_2b_2|}{(1-a^2_1)(1-a^2_2)}. \label{nzQ}
\ee
The right hand side of the above quantity is nonzero if and only if the state has a nonzero discord from Alice to Bob
as well as from Bob to Alice \cite{DVB10} for the two-qubit states with either Alice's or Bob's marginal being maximally mixed, i.e., $\vec{a}=0$ or $\vec{b}=0$.

Therefore, we arrive at the following:
\vspace*{12pt}
\noindent
\begin{proposition}
	Consider the set of measurements which can be used to demonstrate the maximal violation of the two setting linear steering inequality by the maximally entangled two-qubit state. For such measurements, a given two-qubit state gives rise to nonzero $Q$  if and only if it
	is neither a classical-quantum state nor a quantum-classical state, with one of the marginals maximally  mixed. 
	\label{prop2}
\end{proposition}
\vspace*{12pt}
\noindent

Let us now analyze which nonzero discord states having nonmaximally mixed marginals on both the sides can lead 
to nonzero $Q$.
Consider the two-qubit nonzero discord state 
with $a_1=a_2=b_1=b_2=c_1=c_2=1/2$ and $a_3=b_3=0=c_3=0$. For this state, 
$Q$ as given by Eq. (\ref{nzQ}) takes the value $0$. In Ref. \cite{DBD+17}, it has been demonstrated that such separable state
can be used to demonstrate superunsteerability with minimum dimension of the shared being equal to $3$.

Let now demonstrate that $Q=0$ for the above such superunsteerable state even if Alice and Bob perform measurements which optimize the violation of the linear steering 
inequality \cite{CA16}.
For this choice of measurements which have the directions given by 
$\hat{u}_{0}=\frac{c_1\hat{x}+c_2\hat{y}}{\sqrt{c^2_1+c^2_2}}$
and $\hat{u}_{1}=\frac{c_1\hat{x}-c_2\hat{y}}{\sqrt{c^2_1+c^2_2}}$ on Alice's side and  $\vec{T}_{0}=\frac{\hat{x}+\hat{y}}{\sqrt{2}}$
and $\vec{T}_{1}=\frac{\hat{x}-\hat{y}}{\sqrt{2}}$ on Bob's side, $Q$ for the two-qubit states (\ref{can02q})
with $|c_1|\ge |c_2| \ge |c_3| $ has been evaluated to be of the form given by,
\begin{align}
	&Q=  \left|\frac{4 c_1 c_2 (a_2 b_2 c_1 + a_1 b_1 c_2 - c_1 c_2)}{a_1^4 c_1^4 + (-2 + a_2^2 c_2^2)^2 - 2 a_1^2 c_1^2 (2 + a_2^2 c_2^2)}\right|. \label{nzQ1}
\end{align}
Note that $Q$ as given above takes the value zero for the above-mentioned superunsteerable state with $a_1=a_2=b_1=b_2=c_1=c_2=1/2$ and $a_3=b_3=0=c_3=0$. 

In fact, it can be checked that when Alice and Bob perform  arbitrary POVMs given by Eqs. (\ref{POVMalice}) and (\ref{POVMbob}), respectively, on the above-mentioned superunsteerable state with $a_1=a_2=b_1=b_2=c_1=c_2=1/2$ and $a_3=b_3=0=c_3=0$, one will get $Q=0$.  Thus, we arrive at the following  proposition.
\vspace*{12pt}
\noindent
\begin{proposition}
	There exists superunsteerable two-qubit states which cannot be used to demonstrate nonzero $Q$. Such two-qubit states have nonmaximally mixed marginals for both the sides.
	\label{proposition3}
\end{proposition}
\vspace*{12pt}
\noindent

Thus, we identify a new classification of two-qubit states based on the fact that a discordant state may or may not be superunsteerable and a superunsteerable state may or may not be used for selftesting quantum information processing (QIP) (see fig. \ref{Fig:Hirarchy}).
\begin{figure}[h!]
	\begin{center}
		\includegraphics[width=7.5cm]{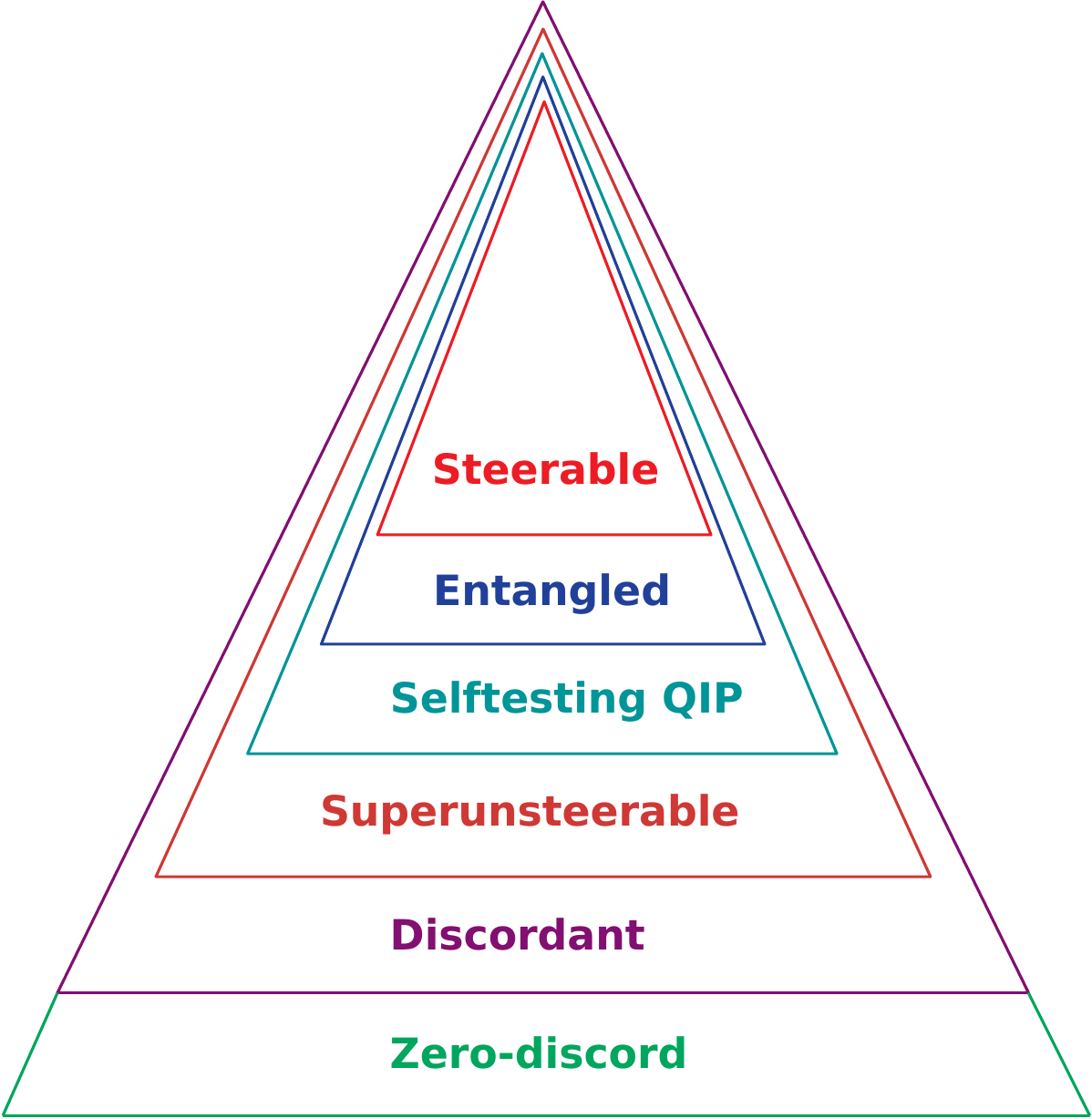}
	\end{center}
	\caption{Hierarchy of correlations in two-qubit states. In the two-setting steering scenario 
		that we consider in this work, nonzero discord (discordant) states can be classified into 
		states which do not demonstrate superunsteerability, 
		superunsteerable states useless for self-testing QIP and useful for self-testing QIP. 
		\label{Fig:Hirarchy}}
\end{figure}

\section{Certifying superlocality}

\subsection{Prepare and measure scenario with correlated devices}

\begin{figure}[h!]
	\begin{center}
		\includegraphics[width=7.5cm]{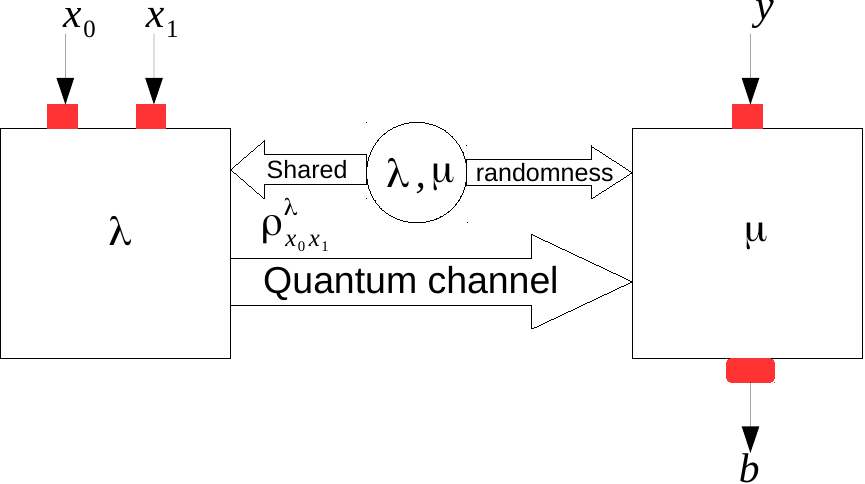}
	\end{center}
	\caption{Prepare and measure  scenario with correlated devices. 
		Here, $x_0,x_1,y,b \in \{0,1\}$ and $\rho^{\lambda}_{x} \in \mathcal{B}(\mathbb{C}^{2})$.
		\label{Fig:PMSR}}
\end{figure}

Let us consider a SDI scenario as in Fig. \ref{Fig:PMSR} which is 
the P$\&$M scenario as in Fig. \ref{Fig:PMNoSR} but with the correlated devices, i.e., $p(\lambda, \mu)\ne q_\lambda \cdot r_\mu$. 
In such a SDI scenario where the devices are not independent, it has been demonstrated that 
the set of distributions $\{p(b|x_0x_1,y)\}$ that achieve the maximal value $1$ for the witness $W$ as given by Eq. (\ref{Wit}) 
can also achieved by a classical bit \cite{PB11}.
Thus, in the presence of shared randomness, the P$\&$M scenario cannot be used to certify the randomness by observing nonzero $W$.

In the context of scenario as in Fig. \ref{Fig:PMSR},
quantum advantage of $2$-to-$1$ random-access code was studied \cite{ALM+08} 
and a linear dimension witness was derived \cite{GBH+10}. 
In  a $2$-to-$1$ quantum random-access code which is implemented using the protocol given in Fig. \ref{Fig:PMSR},
Bob's goal is to guess Alice's $y$th bit.
The average success probability of Bob to guess Alice's $y$th bit is given by 
\be
P_{B}:=\frac{1}{8}\sum_{x_0,x_1,y}P(b=x_y|x_0x_1,y).
\ee
If Alice sends a classical bit to Bob, the optimal
average  success probability is upper bounded by $3/4$. On the other hand, if Alice sends a
qubit to Bob, then the average success probability of Bob can beat this classical bound.
The quantum strategy  that gives the optimal average success probability of $\frac{1}{2}\left(1+\frac{1}{\sqrt{2}}\right)$ is given by
\begin{align}
	\rho_{00}&=\frac{1}{2}\left(\I+\frac{\sigma_x+\sigma_y}{\sqrt{2}}\right)   \nonumber  \\
	\rho_{01}&=\frac{1}{2}\left(\I+\frac{\sigma_x-\sigma_y}{\sqrt{2}}\right)   \nonumber  \\
	\rho_{10}&=\frac{1}{2}\left(\I-\frac{\sigma_x-\sigma_y}{\sqrt{2}}\right)   \nonumber  \\
	\rho_{11}&=\frac{1}{2}\left(\I-\frac{\sigma_x+\sigma_y}{\sqrt{2}}\right)   \label{OPTprep}
\end{align}
and 
\begin{align}
	M_0&=\sigma_x, \quad     M_1=\sigma_y.  \label{OPTmeas}
\end{align}

In Ref. \cite{PB11}, a relationship between the average success probability of $2$-to-$1$
random-access code and the linear dimension witness has been demonstrated. 
The  linear dimension witness is given by 
\begin{align}
	W_{L}&:=P(0|00,0)+P(0|00,1)+P(0|01,0)-P(0|01,1) \nonumber \\
	&-P(0|10,0)+P(0|10,1)-P(0|11,0)-P(0|11,1) \le 2. \label{LDW}
\end{align}
By assuming qubit dimension, the violation of the linear dimension witness inequality (\ref{LDW}) certifies
the quantumness of preparations and measurements. The quantum strategy as given by Eqs. (\ref{OPTprep}) and (\ref{OPTmeas})
gives the maximal quantum violation of $2\sqrt{2}$ for the inequality given by Eq. (\ref{LDW}).
This is related to the maximal quantum violation of the CHSH inequality \cite{CHS+69} as discussed in the 
next section. 
The average success probability of $2$-to-$1$ random-access code is related to the linear dimension witness as  
$P_{B}=\frac{W_L+4}{8}$. From this relationship, it follows that the violation of the  inequality $P_{B} \le 3/4$, 
which certifies the quantum advantage of the random-access code, implies
the violation of the inequality given by (\ref{LDW})  and vice versa.
In the P$\&$M scenario as in Fig. \ref{Fig:PMSR}, the quantum advantage of the random-access codes implied by the violation of the  inequality 
$P_{B} \le 3/4$ or the quantumness as certified by the violation of the linear dimension witness inequality  has been used 
to demonstrate secure quantum key distribution \cite{PB11} and genuine randomness generation \cite{LPY+12}.

\subsection{Relating the quantumness of sequential correlations with that of spatial correlations}

\begin{figure}[h!]
	\begin{center}
		\includegraphics[width=7.5cm]{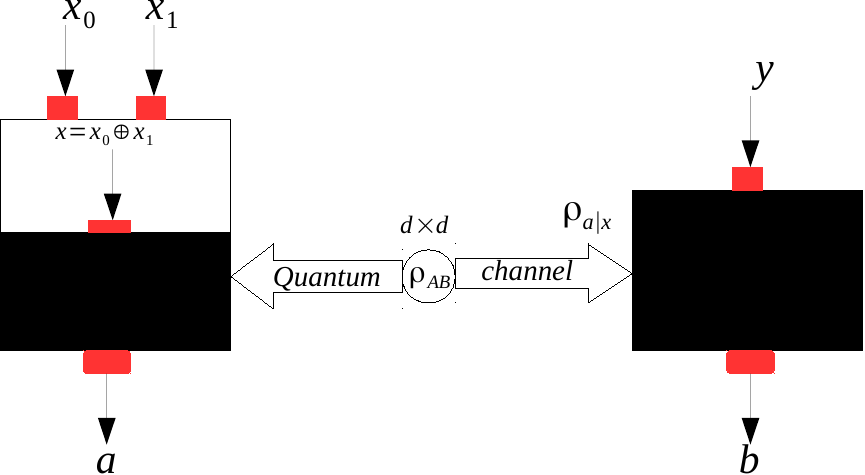}
	\end{center}
	\caption{A device-independent scenario  where a shared bipartite state acts as the quantum channel
		to implement the $2$-to-$1$ random-access code. Here, $\{\rho_{0|0}=\rho_{00}, \rho_{1|0}=\rho_{11}, 
		\rho_{0|1}=\rho_{01}, \rho_{1|1}=\rho_{10}\}$.
		\label{Fig:DI}}
\end{figure}
To observe the quantumness of sequential correlations in the
above scenario depicted in Fig. \ref{Fig:PMSR} as the quantumness of spatial correlations, 
let us now consider the scenario as in Fig. \ref{Fig:DI}, which is a fully device-independent scenario 
where Alice and Bob share a bipartite quantum state and
performs  two dichotomic black box measurements.
Here we consider the fully device-independent scenario since shared randomness is used in the 
corresponding prepare and measure scenario.
In this fully device-independent scenario, let us observe the spatial correlations through the set of
conditional probabilities $\{p(b|a;x,y)\}$ given by
\be
p(b|a;x,y)=\frac{\tr  \left( M^{(1)}_{a|x} \otimes M^{(2)}_{b|y} \rho_{AB} \right) }{\tr \left( M^{(1)}_{a|x} \otimes \I  \rho_{AB} \right) }
\ee
where $\rho_{AB}$ is a bipartite quantum state of arbitrary local Hilbert space dimension
and $M^{(1)}_{a|x}$ and $M^{(2)}_{b|y}$ are the measurement operators of Alice and Bob's measurements respectively.
The scenario in Fig. \ref{Fig:DI} supplemented with one bit of classical communication from Alice to Bob  can be used to implement  the $2$-to-$1$ random-access code  \cite{PZ10}.
Therefore, the violation of the  inequality $P_{B} \le 3/4$ which can be used to certify the quantum advantage of  the $2$-to-$1$  
random-access code assisted by shared bipartite state implies that the set of joint behaviours $\{P(a,b|x,y)=p(b|a;x,y)p(a|x)\}$ 
violates a Bell inequality.
Thus, we can conclude that the quantumness in the P$\&$M scenario with the correlated devices 
is manifested in the form of Bell-nonlocal  
correlations in the corresponding scenario where a shared bipartite state is used as  the quantum channel. Further,
observing the violation of the inequality (\ref{LDW}) in the context of producing the distributions  $\{P(b|x_0x_1,y)\}$
using a shared bipartite state witnesses Bell nonlocality. Thus, the linear dimension witnesess  provide
Bell inequalities in terms of the conditional probabilities \footnote{In Refs. \cite{PKM14,MBM15}, steering inequalities
	and separability inequalities are derived respectively using the conditional probabilities. To our knowledge,
	Bell inequalities in terms of the conditional probabilities have not been studied so far.}.

\subsection{Quantum correlations beyond Bell-nonlocality as a manifestation of quantumness of sequential correlations}
In a bipartite Bell scenario, superlocality \cite{DW15} is defined as follows.

\vspace*{12pt}
\noindent
{\bf Definition~2.}
Suppose we have a quantum state in $\mathbb{C}^{d_A}\otimes\mathbb{C}^{d_B}$
and measurements which produce a local bipartite box $\{p(a,b|x,y)\}$.
Then, superlocality holds iff there is no decomposition of the box in the form,
\begin{equation}
	p(a,b|x,y)=\sum^{d_\lambda-1}_{\lambda=0} p_{\lambda} p(a|x, \lambda) p(b|y, \lambda) \hspace{0.3cm} \forall a,x,b,y,
\end{equation}
with dimension of the shared randomness/hidden variable $d_\lambda\le$ min($d_A$, $d_B$).  Here $\sum_{\lambda} p_{\lambda} = 1$, $p(a|x, \lambda)$ and $p(b|y, \lambda)$ denotes arbitrary probability distributions arising from LHV $\lambda$ ($\lambda$ occurs with probability $p_{\lambda}$).
\vspace*{12pt}
\noindent

Superlocality provides an operational characterization to the quantumness of local  boxes \cite{JAS17,JDS+18}.

Previously, it was shown that the quantity $Q$ given by Eq. (\ref{QC}) of any two-qubit state is linked with superunsteerability in the scenario depicted in Fig. \ref{Fig:SDI}. Now, we will show that $Q$ can be linked with superlocality as well in the present context.
\vspace*{12pt}
\noindent
\begin{proposition}
	Suppose any nonzero value of the quantity $Q$  given by Eq. (\ref{QC})  arises from a two-qubit state in the context of the scenario depicted in Fig. \ref{Fig:DI}. Then it implies the presence of  superlocality or Bell-nonlocality in the correlation $\{P(a,b|x,y)\}$. 
\end{proposition}
\vspace*{12pt}
\noindent
{\bf Proof.} Note that in our device-independent scenario, 
any local box having the following LHV-LHV decomposition,
\begin{equation}
	p(a,b|x,y)=\sum^{d_\lambda-1}_{\lambda=0} p_{\lambda} p(a|x, \lambda) p(b|y, \lambda) \hspace{0.3cm} \forall a,x,b,y,
\end{equation}
can be reproduced by performing appropriate measurements on the quantum-classical state of
the form $\sum^{d_\lambda-1}_{\lambda=0} p_\lambda \rho_{\lambda} \otimes \ketbra{\lambda}{\lambda}$, or on the classical-quantum state of
the form $\sum^{d_\lambda-1}_{\lambda=0} p_\lambda \ketbra{\lambda}{\lambda} \otimes \rho_{\lambda}$
where $\{\ketbra{\lambda}{\lambda}\}$ forms an orthonormal basis in  $\mathbb{C}^{d_\lambda}$, with $d_\lambda \le 4$ \footnote{Here the dimension of the hidden variable is upper bounded by $4$
	since any local correlation corresponding to this scenario can be simulated by shared classical randomness of dimension $d_\lambda \le 4$ \cite{DW15}.} This implies that any local box arising from a two-qubit state that requires a hidden variable of dimension $d_\lambda =2$
for providing a LHV-LHV model (i.e., any local box arising from a two-qubit state that is not superlocal)
can be simulated by a two-qubit state which admits the form of the classical-quantum  state or quantum-classical state. Now, such state cannot give non-zero $Q$ (follows from the proof of proposition \ref{propo01}).
Thus, for any local box arising from a two-qubit state that is not superlocal, $Q=0$. On the other hand, $Q>0$ produced from a two-qubit state certifies that the box does not arise from a classical-quantum or quantum-classical state. Hence, that box either
requires the hidden variable of dimension $d_\lambda>2$ for providing a LHV-LHV model (i.e., the box is superlocal) or that box is Bell-nonlocal  $\square\,$.

As discussed in Sec. \ref{P&MID}, nonzero $Q$ provides selftesting random number generation through 
witnessing superunsteerability/steering in a SDI way. It is readily seen that 
nonzero $Q$ in the context of above mentioned Bell scenario also 
provides selftesting random number generation as given  by Eq. (\ref{STQRNG}) through 
witnessing superlocality or Bell-nonlocality in a SDI way.

Note that for the measurements that has been used to demonstrate superunsteerability of any two-qubit
nonzero discord state in Proposition \ref{prop2}, the correlations arising from certain two-qubit states violate 
the two-setting linear steering inequality and the correlations arising from maximally entangled two-qubit states violate 
the two-setting linear steering inequality maximally \cite{CA16}. We now demonstrate that for the measurements that 
give rise to the maximal Bell-CHSH inequality violation  or maximal quantum advantage of $2$-to-$1$ random-access code by the maximally entangled state, the two-qubit states given by Eq. (\ref{can02q})
with $|c_1|\ge |c_2| \ge |c_3| $ give rise to non-zero $Q$. 
Let Alice performs measurements along the directions $\hat{u}_{0}=\frac{\hat{x}+\hat{y}}{\sqrt{2}}$
and $\hat{u}_{1}=\frac{\hat{x}-\hat{y}}{\sqrt{2}}$ and Bob performs the measurement along the directions $\vec{T}_{0}=\hat{x}$
and $\vec{T}_{1}=\hat{y}$.
For this choice of measurement directions which can be used to provide the maximal quantum advantage 
of $2$-to-$1$ random-access codes by the maximally entangled state \cite{PZ10},  the two-qubit states given by Eq. (\ref{can02q}) with $|c_1|\ge |c_2| \ge |c_3| $ give rise to the following expression of $Q$,
\begin{equation}
	Q = \frac{2|(c_1-a_1b_1)(c_2-a_2b_2)-a_1b_1a_2b_2|}{a_1^4 + (-2 + a_2^2)^2 - 2 a_1^2 (2 + a_2^2)}. \label{nzQ2}
\end{equation}
The right hand side of the above quantity is nonzero if and only if the state has a nonzero discord from Alice to Bob
as well as from Bob to Alice \cite{DVB10} provided that either Alice's or Bob's marginal of that two-qubit state is being maximally mixed, i.e., $\vec{a}=0$ or $\vec{b}=0$. Therefore, we arrive at the following.
\vspace*{12pt}
\noindent
\begin{proposition}
	$Q$ is nonzero for a given two-qubit state if and only if  it is superlocal and has one of the marginals maximally mixed. On the other hand, in general, a nonzero value of $Q$ provides sufficient certification of superlocality.
\end{proposition}
\vspace*{12pt}
\noindent
Thus, we identify a new classification of superlocal states according to whether it can be used for self-testing quantum information processing (QIP) or not (see fig. \ref{Fig:Hirarchylocal}).
\begin{figure}[h!]
	\begin{center}
		\includegraphics[width=7.5cm]{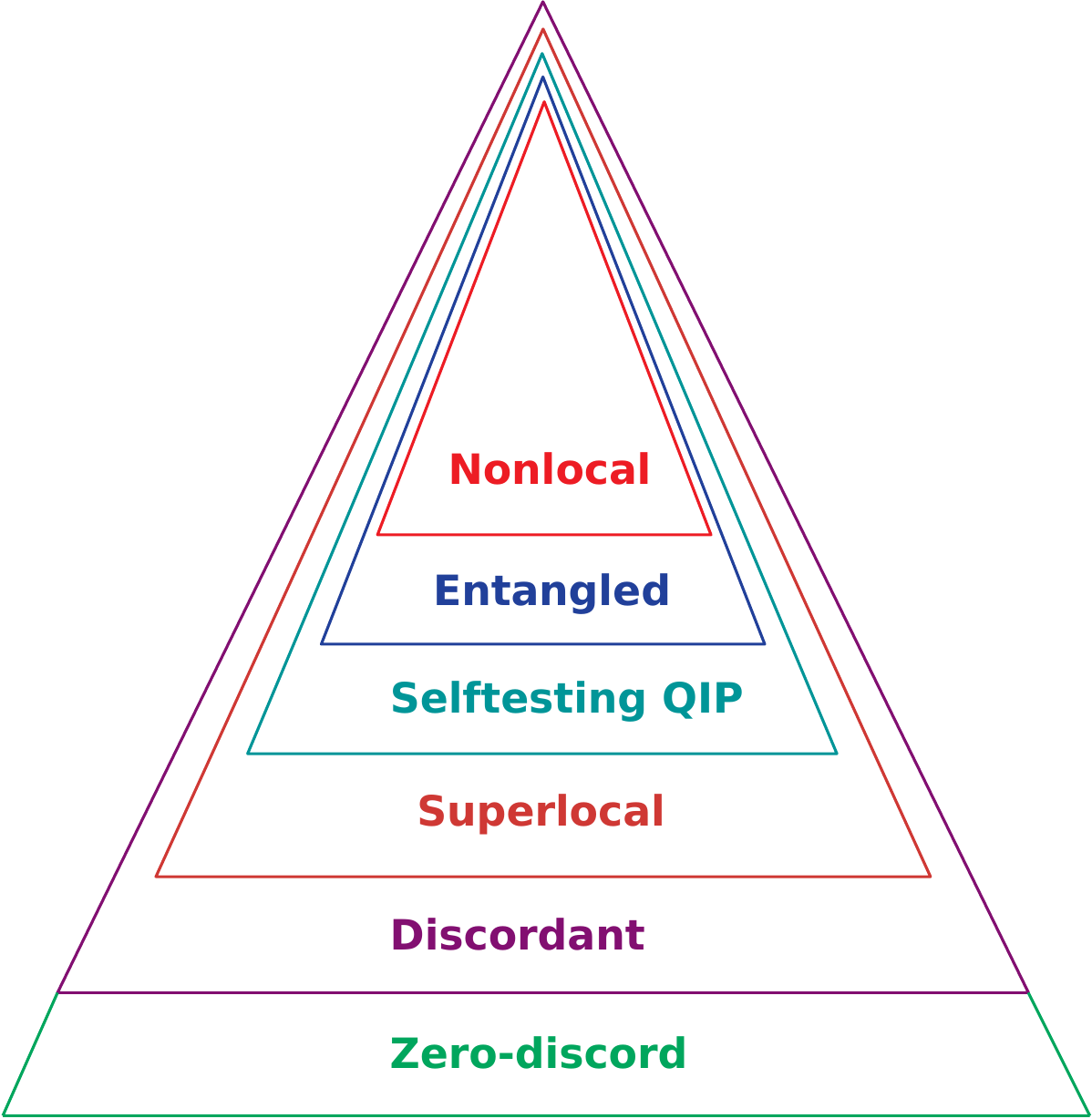}
	\end{center}
	\caption{ Hierarchy of correlations in two-qubit states.
		\label{Fig:Hirarchylocal}}
\end{figure}

Here it should be noted that for the set of quantum correlations that gives $Q$ value  as given by Eq. (\ref{nzQ2}), $Q$ is upper bounded  by $Q \le 1$. On the other hand,
for the Popescu-Rohrlich boxes  which are nonsignaling correlations, but are stronger than quantum correlations violating
a Bell-CHSH inequality to its algebraic maximum of $4$ \cite{BLM+05}, $Q$ takes the algebraic maximum of $2$.

\subsection{Superlocality as resource for $2$-to-$1$ and $3$-to-$1$ random-access codes}
In Ref. \cite{BP14}, Bobby and Paterek  studied quantum advantage of  $2$-to-$1$ and $3$-to-$1$ random-access codes assisted by 
two correlated qubits in the presence of two-bits of shared randomness. As a figure merit of the task,  
Bobby and Paterek considered  worst-case success probability of Bob's correct guess on Alice's $y$th bit defined as  
${{P}_{{\rm min} }}={{{\rm min} }_{x_0,x_1,y}} p (b=x_y|x_0x_1,y)$. 
Bobby and Paterek derived the following inequalities:
A classical $n$-to-$1$ random-access code assisted with two bits from a common source has 
\be
{{P}_{{\rm min} }}\leqslant \frac{1}{2} \quad \texttt{if}  \quad n > 2; \label{SL1}
\ee
\be
{{P}_{{\rm min} }}\leqslant \frac{2}{3} \quad \texttt{if}  \quad n = 2;  \label{SL2}
\ee
and 
\be
{{P}_{{\rm min} }}\le \frac{1}{2} \label{SL3}
\ee
for all $n > 1$ if the assisting bits have maximally mixed marginals for Bob. \label{Ineq}
By using these inequalities, Bobby and Paterek  have shown that even certain separable two-qubit state may become useful for implementing   $2$-to-$1$ and $3$-to-$1$ quantum random-access codes.

Note that for $n=2$ in Eqs. (\ref{SL1}), (\ref{SL2}) and (\ref{SL3}), the 
inequalities correspond to the Bell-Clauser-Horne-Shimony-Holt (Bell-CHSH) scenario \cite{CHS+69} while for $n=3$, the 
inequalities are related to the bipartite Bell scenario corresponding to Gisin’s elegant Bell inequality \cite{Gis09}
\footnote{In Ref. \cite{GP18}, in the context of  $n$-bit  parity-oblivious multiplexing task,
	the suitable Bell expressions corresponding to quantum success probability of this task has been identified. The Bell inequalities corresponding
	to these Bell expressions, in turn, are related to quantum advantage of $n$-to-$1$ random-access codes.}.
The violation of the inequalities given by Eqs. (\ref{SL1}) and (\ref{SL2}) by the local boxes 
implies that the simulation of these boxes by using shared classical randomness requires the hidden variable
of dimension $d_\lambda>2$. Therefore, these inequalities serve as sufficient certification of superlocality.

For the protocol given in Ref. \cite{BP14} which implements the $2$-to-$1$ quantum random-access code using the Bell-diagonal states
which are the two-qubit states given by Eq. (\ref{can02q}) with $\vec{a}=\vec{b}=0$, 
the separable state with $c_1=c_2=1/2$ and $c_3=0$
optimizes the quantum advantage within the separable Bell-diagonal states. In this case, the protocol prepares the following
conditional states on Bob's side,
\begin{align}
	\rho_{0|0}&=\rho_{00}=\frac{1}{2}\left(\I+\frac{\sigma_x+\sigma_y}{2\sqrt{2}}\right)   \nonumber  \\
	\rho_{0|1}&=\rho_{01}=\frac{1}{2}\left(\I+\frac{\sigma_x-\sigma_y}{2\sqrt{2}}\right)   \nonumber  \\
	\rho_{1|1}&=\rho_{10}=\frac{1}{2}\left(\I-\frac{\sigma_x-\sigma_y}{2\sqrt{2}}\right)   \nonumber  \\
	\rho_{1|0}&=\rho_{11}=\frac{1}{2}\left(\I-\frac{\sigma_x+\sigma_y}{2\sqrt{2}}\right)   
\end{align}
and Bob performs the following measurements:
\begin{align}
	M_0&=\sigma_x, \quad     M_1=\sigma_y 
\end{align}
It has been checked that the above preparations and measurements do not violate the dimension witness inequality given by Eq. (\ref{LDW}).
Thus, there exist quantum strategies in the P$\&$M scenario  
which have quantumness and cannot be used to provide quantum advantage for the random-access code.  However,  the spatial correlations 
realized using such  quantum strategies  in the scenario as in Fig. \ref{Fig:DI} violate  the inequality
given by Eq. (\ref{SL2}) with $n=2$ and give rise to nonzero $Q$. Therefore, these quantum strategies provides advantage for the random-access codes
in the presence of limited shared randomness, i.e., two bits of shared classical randomness. 
Thus, quantumness of certain sequential correlations gets manifested in the form of superlocal
correlations which represent the stronger than classical correlations in the presence of limited shared randomness.

\section{Discussions and Conclusions}
Developing protocols for self-testing quantum information processing is important for realizing quantum technologies 
for real-life applications. Self-testing quantum information protocols require certification of quantum resources 
either in a device-independent or in a semi-device-independent way. Device-independent quantum information protocols, 
which require genuine demonstration of Bell inequality violation, certifies quantumness, i.e., quantum entanglement
in the presence of adversaries who can implement local-hidden-variable models using an infinite amount of shared randomness or finite
shared randomness with sufficiently large dimension  \cite{BHQ+15}. On the other hand, semi-device-independent
quantum information protocols based on quantum steering certifies quantum entanglement in the presence of adversaries
who can implement local-hidden-state models in which only the untrusted parties have access to a classical random variable
of an infinite amount or of finite but sufficiently large dimension. 
In the context of semi-device-independent protocols where only the local Hilbert-space dimension of each party is trusted or fixed, 
quantum entanglement is certified in the presence of adversaries  who have access to restricted amount of finite shared randomness \cite{GBS16}. In such realistic quantum information protocols  in the presence of restricted amount of shared randomness, separable states can also simulate stronger-than classical correlations \cite{JAS17,JDS+18,DBD+17}.
Given the background that separable states are also useful for quantum information processing in certain circumstances, 
it is relevant to develop self-testing quantum information processing tasks based on quantumness of correlations beyond entanglement in the presence of limited shared randomness. 

With this motivation,
in this work, we have studied the equivalence between the quantumness of sequential correlations in the prepare 
and measure scenario with independent devices, which implements quantum random number generation, and the corresponding scenario which replaces quantum communication by sharing correlated particles.
In this context, we have demonstrated that quantumness of sequential correlations gets manifested as 
superunsteerable correlations which are stronger than classical correlations beyond entanglement in the presence 
of limited shared randomness. We have introduced an experimentally measurable quantity to bound
the genuine randomness generation in the scenario using shared correlated particles as resource.
We have shown that this quantity provides certification of superunsteerable correlations. Moreover, our certification of superunsteerability provides necessary and sufficient certification of any two-qubit state  which is neither 
a classical-quantum state nor a quantum-classical state and has one of the marginals maximally in a semi-device-independent way. Finally, we study the connection of our certification of superunsteerability with certification of superlocality.

In the case of 
prepare and measure scenario with correlated devices, the quantumness of sequential correlations
providing quantum advantage for random-access codes 
gets manifested as Bell nonlocal correlations in the corresponding scenario assisted by
shared correlated particles. We have demonstrated that when certain sequential correlations 
which have quantumness but do not violate the dimension witness inequality are realized as spatial correlations, 
they violate an inequality detecting superlocal correlations  which are stronger than the classical correlations
in the presence of limited shared randomness.

In Ref. \cite{DBD+17}, it has been demonstrated that the superunsteerable (superlocal) states in the above SDI scenario can be classified into (i) quantum-quantum states which demonstrate super-unsteerability (superlocality) with unsteerable (local) boxes having minimum hidden variable dimension $3$, and (ii) quantum-quantum states which demonstrate super-unsteerability (superlocality) with unsteerable (local) boxes having minimum hidden variable dimension $4$. Our example presented in proposition \ref{proposition3} indicates that certain superunsteerable states which have rank two and belong to the class (i) may not lead to nonzero $Q$. Hence, it would be interesting to check whether the quantity $Q$ proposed in this work  to certify superlocaliy and superunsteerability  provides necessary and sufficient certification 
of superlocal or superunsteerable states belonging to class (ii) with unsteerable or local boxes respectively having minimum
hidden variable dimension $4$.

It would be interesting to study whether this quantity
is upper bounded by $1$ for quantum correlations.
This will be useful  to discriminate quantum and post-quantum correlations \cite{PS15} since for the Popescu-Rohrlich boxes, this quantity
takes the algebraic maximum of $2$. 
It would be interesting to generalize the present work to quantum correlations with more number of outputs or inputs
as well as with more number of parties.
In Refs. \cite{LTB+14,PCS+15}, it has been demonstrated that genuine randomness can be certified in the presence of
local-hidden-state models. It would be interesting to study implications of the present work to these previous works. Finally, in the light of the present work, we plan to formulate a resource theory of 
superlocality and superunsteerability just like the resource theory of Bell nonlocality 
and quantum steering, respectively \cite{CG19, WSS+19,ZSHS23}.

\section*{Acknowledgement}
 C.J. acknowledges  the  financial support
from the Ministry of Science and Technology, Taiwan (Grant
No. MOST 108-2811-M-006-516). We would like to thank  Prof. Yeong-Cherng Liang and Prof. Paweł Horodecki for helpful discussions. DD acknowledges the Royal Society (United Kingdom) for the support through the Newton International Fellowship (NIF$\backslash$R$1\backslash212007$).

\bibliography{PM}

\end{document}